%% file: main.tex
\documentclass[12pt,a4paper,twoside]{report}

\usepackage[pdfborder={0 0 0}]{hyperref}  
\usepackage[vmargin=20mm,hmargin=25mm]{geometry}  
\usepackage{graphicx} 
\usepackage{parskip}  
\usepackage{setspace} 
\usepackage{refcount} 
\usepackage{upquote}  
\usepackage{verbatim}
\newenvironment{code}{\verbatim}{\endverbatim}

\usepackage{graphicx}
\usepackage{amsmath}
\usepackage{amssymb}
\usepackage{amsfonts}
\usepackage{amsthm}
\usepackage{mathrsfs}
\usepackage{stmaryrd}
\usepackage{quiver}

\usepackage{newunicodechar}
\newunicodechar{ω}{\ensuremath{\omega}}
\newunicodechar{≤}{\ensuremath{\le}}
\newunicodechar{⊑}{\ensuremath{\sqsubseteq}}
\newunicodechar{≡}{\ensuremath{\equiv}}
\newunicodechar{≠}{\ensuremath{\ne}}
\newunicodechar{⊔}{\ensuremath{\sqcup}}
\newunicodechar{₂}{\ensuremath{{}_2}}
\newunicodechar{₃}{\ensuremath{{}_3}}
\newunicodechar{ℕ}{\ensuremath{\N}}
\newunicodechar{□}{\ensuremath{\square}}
\newunicodechar{∙}{\ensuremath{\cdot}}
\newunicodechar{⊥}{\ensuremath{\bot}}
\newunicodechar{∀}{\ensuremath{\mathnormal\forall}}

\newtheorem{definition}{Definition}[chapter]
\newtheorem{conjecture}[definition]{Conjecture}

\newtheorem{axiom}[definition]{Axiom}
\newtheorem{theorem}[definition]{Theorem}
\newtheorem{lemma}[definition]{Lemma}

\newcommand{\prn}[1]{\left(#1\right)}

\newcommand{\set}[1]{\left\{#1\right\}}
\newcommand{\sem}[1]{\left\llbracket#1\right\rrbracket}

\newcommand{\id}{\operatorname{id}}
\newcommand{\rec}{\operatorname{rec}}
\newcommand{\elim}{\operatorname{elim}}

\newcommand{\dom}{\operatorname{dom}}

\newcommand{\I}{\mathbb{I}}
\newcommand{\E}{\mathbb{E}}
\newcommand{\N}{\mathbb{N}}

\newcommand{\U}{\mathcal{U}}
\newcommand{\Op}{\mathcal{O}}

\newcommand{\w}{\omega}
\newcommand{\wb}{\overline{\omega}}

\newcommand{\parto}{\rightharpoonup}

\newif\ifsubmission 

\title{Topology in Synthetic Domain Theory and its Formalisation in Agda}
\author{Runze Xue}
\date{June 2025}
\newcommand{\candidatenumber}{9661I}
\newcommand{\college}{Girton College}
\newcommand{\coursethe}{Master of Philosophy in Advanced Computer Science}
\newcommand{\coursefor}{the Master of Philosophy in Advanced Computer Science}

\begin{document}

\begin{sffamily} 

\begin{titlepage}
\makeatletter

\hspace*{-14mm}\includegraphics[width=65mm]{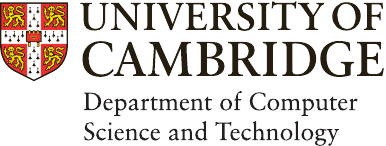}

\ifsubmission

\begin{Large}
\vspace{20mm}
Research project report title page

\vspace{35mm}
Candidate \candidatenumber

\vspace{42mm}
\textsl{``\@title''}

\end{Large}

\else

\begin{center}
\Huge
\vspace{\fill}

\@title
\vspace{\fill}

\@author
\vspace{10mm}

\Large
\college
\vspace{\fill}

\@date
\vspace{\fill}

\end{center}

\fi

\vspace{\fill}
\begin{center}
Submitted in partial fulfilment of the requirements for the\\
\coursethe
\end{center}

\makeatother
\end{titlepage}

\newpage

Total page count: \pageref{lastpage}

\makeatletter
\@tempcnta=\getpagerefnumber{lastcontentpage}\relax%
\advance\@tempcnta by -\getpagerefnumber{firstcontentpage}%
\advance\@tempcnta by 1%
\xdef\contentpages{\the\@tempcnta}%
\makeatother

Main chapters (excluding front-matter, references and appendix):
\contentpages~pages
(pp~\pageref{firstcontentpage}--\pageref{lastcontentpage})

Main chapters word count: 6190

Methodology used to generate that word count:

\begin{quote}
\begin{verbatim}
https://app.uio.no/ifi/texcount/online.php

File: report.tex
Encoding: utf8
Sum count: 6190
Words in text: 6190
Words in headers: 103
Words outside text (captions, etc.): 127
Number of headers: 25
Number of floats/tables/figures: 1
Number of math inlines: 594
Number of math displayed: 85
\end{verbatim}
\end{quote}

\end{sffamily}

\vspace{\fill}
\onehalfspacing
\textbf{\Huge Declaration}
\vspace{40pt}

I,
\makeatletter\ifsubmission
\candidatenumber,
\else
\@author\ of \college,
\fi\makeatother
being a candidate for \coursefor, hereby declare that this report and
the work described in it are my own work, unaided except as may be
specified below, and that the report does not contain material that
has already been used to any substantial extent for a comparable
purpose.
In preparation of this report, I adhered to the
\href{https://www.cst.cam.ac.uk/files/ai_policy.pdf}{Department of
Computer Science and Technology AI Policy}. I am content for
my report to be made available to the students and staff of the
University.



\ifsubmission\else
\bigskip
\textbf{Signed: Runze Xue}
\fi

\bigskip
\textbf{Date: 11th June, 2025}
\vspace{\fill}

\input{chapters/abstract}

\ifsubmission\else

\input{chapters/acknowledgement}

\fi
\cleardoublepage 

\tableofcontents

\label{firstcontentpage} 

\input{chapters/1-intro}

\input{chapters/2-background}

\input{chapters/3-interval}

\input{chapters/4-phoa}

\input{chapters/5-omega}

\input{chapters/6-complete}

\input{chapters/7-conclusion}

\label{lastcontentpage} 

\bibliographystyle{plain}
\bibliography{ref}










\appendix

\input{chapters/a-proof}

\label{lastpage}

\end{document}

%% file: chapters/abstract.tex

\chapter*{Abstract}

This project investigates the Phoa principle in synthetic domain theory (SDT), and provides a generalisation to the transfinite cases. The Phoa principle plays a pivotal role in SDT by illustrating how the paths give the information order on the interval type and other algebraic structures in SDT. The project defines the dual simplices and spines and introduces the concept of sobriomorphisms, which contributes to a new interpretation of the Phoa principle and its generalisations. Finally, the project proposes a hypothetical completeness theorem that may unify the Segal completeness and the chain completeness in SDT based on investigations on the Phoa principle in the project.

The project also includes axiomatisation of the interval type in Cubical Agda and the formalised proof for the main theorems.

%% file: chapters/acknowledgement.tex

\chapter*{Acknowledgments}

The author feels obliged to extend his highest gratitude to the project advisor, Dr. Jonathan Sterling. This project would not have been possible without his extraordinary guidance and support, as well as his consistent patience and professionalism during project supervision.

The author would also like to express sincere thankfulness to his friends from the Chinese programming language community: Xu ``Trebor'' Huang, Yinsen ``Tesla'' Zhang, Yue ``Patrick'' Yao and Yuchen ``Alex'' Jiang for providing helpful inspirations and insights on logic and type theory.

Finally, the author wishes to say thanks to Soruto and Mika Yurisaki for their enthusiastic companion and encouragement during the past year.

%% file: chapters/1-intro.tex

\chapter{Introduction}

Domain theory is a powerful tool for dealing with partial functions in the denotational semantics of programming languages. The classical approach assigns each type a domain, which is a poset with directed suprema, and sees functions between these types as homomorphisms between domains \cite{scott_domains_1982,scott_outline_1970}. Meanwhile, the classical approach does not immediately scale to different computation models with different strengths of computability, or models like concurrent computations \cite{nygaard_domain_2004,winskel_powerdomains_1985}.

Synthetic domain theory (SDT) is a synthetic and axiomatic approach to domain theory \cite{fiore_axiomatic_1996,hyland_first_1991,reus_general_1997}. In SDT, the domains and orders are not introduced a priori. Instead, it constructs the theory in a topos \cite{mac_lane_sheaves_1994} with a classifier that classifies a certain class of subsets \cite{hyland_first_1991}. The illustration of partial functions and the construction of domains are then derived from this classifier.

However, defining the orders inside SDT is a non-trivial topic and leads to different choices of axioms that yield different models of SDT \cite{van_oosten_axioms_2000}. One promising approach is inspired by the homotopy type theory (HoTT) \cite{the_univalent_foundations_program_homotopy_2013}. HoTT gives a type theory that treats identity types as paths and uses algebraic topological constructions to give important properties related to identity types, including congruence, functional extensionality, and univalence. Due to the nature of identity types and homotopy, HoTT can be seen as a theory of $\infty$-groupoids.

Inspired by this, Riehl and Schulman gave a variant HoTT known as directed homotopy type theory (directed HoTT) and constructed a synthetic theory of $\infty$-categories \cite{riehl_type_2023}. To restrict a type in directed HoTT to behave like a $\infty$-category, the type must satisfy two completeness properties called Segal and Rezk. By seeing a poset as a thin $1$-category and taking the interval type as the subobject classifier, directed HoTT can therefore be used to axiomatise SDT, and hence we need to find proper axioms to imply that the interval itself is a domain, which includes showing that the interval is Segal and Rezk.

Phoa principle is one property about the classifier \cite{phoa_domain_1991} that is proven to be enough to imply the Segal and Rezk properties for the interval type along with other properties that altogether show that the interval is a poset \cite{sterling_domains_2025}. The Phoa principle has also been generalised to higher-dimensional simplices, and such generalisation is shown to be equivalent to the original Phoa principle \cite{pugh_when_2025}.

This thesis proposes a further generalisation of the Phoa principle and provides verified proofs in Cubical Agda \cite{vezzosi_cubical_2019}. Specifically, the thesis contains the following works:
\begin{enumerate}
    \item Rephrasing the higher Phoa principle by introducing a dual version of simplices.
    \item Generalising the higher Phoa principle to the transfinite cases based on the rephrasing.
    \item Introducing the concept of sobriomorphisms and revealing a possible connection between the Phoa principle and the chain-completeness, a property that makes a synthetic poset a domain.
    \item Axiomatising the interval type in Cubical Agda and proving the classical, higher, and transfinite Phoa principle with concrete constructions in Cubical Agda.
\end{enumerate}

The thesis starts with the algebraic structures on the interval type and illustrates the connection between the partial map classifier and the simplices in Chapter \ref{chapter-interval}. In Chapter \ref{chapter-phoa}, we give a review on the classical and higher Phoa principle with a rephrasing based on the dual simplices and sobriomorphisms. Then in Chapter \ref{chapter-omega}, we generalise the Phoa principle to the transfinite case that relates the simplices and the initial and final (co)algebras defined by the partial map classifier. Finally, in Chapter \ref{chapter-topology}, we make a brief investigation on the relation between the Phoa principle and the completeness properties involved in synthetic domains, and propose a conjectural completeness property that may possibly unify the Segal and the chain completeness. Several theorems and lemmata in the text are accompanied by an Agda object name written in monospace fonts. These refer to the corresponding formalisation code in the Agda codebase.

%% file: chapters/2-background.tex

\chapter{Background}

\section{Domain Theory}

Domain theory, first proposed by Scott \cite{scott_domains_1982,scott_outline_1970}, aims to provide a denotational semantics for languages with general recursions. The core idea in domain theory is to see types as domains, a poset with supremum for all directed subsets, and functions between types are monotonic functions preserving the directed suprema. We hereby introduce some fundamental concepts involved in domain theory. The contents are common results in the study of denotational semantics and are loosely based on the literature by Davey and Priestley \cite{davey_introduction_2002} and the lecture notes by Lennon-Bertrand \cite{lennon-bertrand_denotational_2024}.

\begin{definition}
    A binary relation $\preceq$ on a set $P$ is a \textbf{partial order} if it satisfies the following properties:
    \begin{itemize}
        \item Reflexivity: $\forall x\in P.x\preceq x$.
        \item Transitivity: $\forall x\in P.\forall y\in P.\forall z\in P.x\preceq y\wedge y\preceq z\implies x\preceq z$.
        \item Antisymmetry: $\forall x\in P.\forall y\in P.x\preceq y\wedge y\preceq x\implies x=y$.
    \end{itemize}
    A set equipped with a partial order $(P,\preceq)$ is called a \textbf{poset}.
\end{definition}

\begin{definition}
    Let $(P,\preceq)$ be a poset and $S\subseteq P$ be a subset of $P$. $s$ is an \textbf{upper (lower) bound} of $S$ if $\forall x\in S.s\preceq x$ ($\forall x\in S.x\preceq s$). If the set of all upper (lower) bounds of $S$ has a lower (upper) bound, it is called the \textbf{supremum (infimum)} of $S$, denoted $\sup S$ ($\inf S$). Suprema and infima are unique by definition if they exist.
\end{definition}

\begin{definition}
    Let $(D,\sqsubseteq)$ be a poset. A \textbf{$\omega$-chain} in $D$ is a sequence of elements $(a_n\in D)_{n\in\N}$ that satisfy $a_m\sqsubseteq a_n$ for any $m\le n$. $D$ is called a \textbf{predomain} or a \textbf{$\omega$-cpo} if all $\omega$-chains on $D$ have suprema. If in addition $D$ has a minimum, $D$ is called a \textbf{domain} or a \textbf{$\omega$-cppo}.\footnote{Some literature uses the concept of dcpo instead, requiring all directed subsets to have suprema, which is a stronger condition than the one we used here. However, $\omega$-chain completeness is sufficient in constructing fixed points and is consistent with the synthetic formalizations of domains in Chapter \ref{chapter-topology}.}
\end{definition}

An inspiring example of a domain is the set of partial functions $A\parto B$, where $f\sqsubseteq g$ iff $\dom f\subseteq\dom g$ and $\forall x\in\dom f.f(x)=g(x)$. In other words, $f\sqsubseteq g$ if $g$ is more ``informative'' than $f$ as a function. The order $\sqsubseteq$ in a domain is therefore often called the information order.

\begin{definition}
    A function between posets is \textbf{monotonic} if
    $$\forall x.\forall y.x\preceq y\implies f(x)\preceq f(y)$$
    A function between domains is \textbf{Scott-continuous} (or \textbf{continuous} for short) if it is monotonic and preserves all directed suprema.
\end{definition}

One key result about continuous functions in domain theory is Kleene's fixed point theorem.
\begin{theorem}[Kleene's fixed point theorem]
    Let $f:D\to D$ be a continuous function. $f$ always have a least fixed point given by
    $$\sup_{n\in\N}f^n(\bot)=\sup\set{\bot,f(\bot),f(f(\bot)),\dots}$$
\end{theorem}
The fixed point theorem illustrates the existence of recursion in the sense that the recursion can be constructed as a fixed point in the domain of partial functions using the $Y$-combinator.

\section{Synthetic Domain Theory}

Synthetic domain theory (SDT), as its name suggests, is a synthetic approach to domain theory \cite{hyland_first_1991,reus_general_1997,reus_program_1996}. In SDT, domains are treated as intuitionistic sets, and all functions expressible inside this are by nature continuous.

An important motivation of SDT is that the formalisation in domain theory often involves tedious definitions of the orders on different types and constructions of continuous functions between them. SDT tackles this issue by introducing a special object $\I$\footnote{Other common notations for the classifier include $\mathbb{S}$ and $\Sigma$.} that serves as a subobject classifier classifying a certain class of subsets on which the partial functions are studied. By postulating different axioms about $\I$, we may have different variants of the SDT scaling to different computation models, which is another advantage of SDT.

Meanwhile, SDT provides a type theory for users to formalise the proofs inside the language, and can be incorporated with other type theories, for example, HoTT\cite{the_univalent_foundations_program_homotopy_2013}, which allows users to deal with higher structures inside the theory.

\section{Homotopy Type Theory}

\begin{table}[h]
    \centering
    \begin{tabular}{l|l|l}
    \textbf{Homotopy Theory} & \textbf{Type Theory} & \textbf{Category Theory} \\ \hline
    Space & Type & Category \\ \hline
    Point & Term & Object \\ \hline
    Path & Identity Type & Morphism \\ \hline
    Identity Path & Reflexivity of Equality & Identity Morphism \\ \hline
    Path Composition & Transitivity of Equality & Morphism Composition \\ \hline
    Path Inversion & Symmetry of Equality &  \\ \hline
    Continuous Function & Function & Functor \\ \hline
    & Function Congruence & Functor Mapping Morphisms \\ \hline
    & Function Extensionality & Functor Mapping Objects \\ 
    \end{tabular}
    \caption{A Rosetta Stone between Homotopy Theory, Type Theory, and Category Theory}
    \label{tab-rosetta}
\end{table}

In this thesis, we use the homotopy type theory (HoTT) as our meta-language for theorem proving. HoTT is a type theory in which identity types are treated as paths in homotopy theory \cite{the_univalent_foundations_program_homotopy_2013}. Table \ref{tab-rosetta} lists the correspondence between the key concepts in homotopy theory and type theory. Functions in HoTT by nature satisfy the congruence rule (all functions preserve equality) and the function extensionality (two functions are equal iff their values are equal anywhere). See Section \ref{sect-order} for explicit proofs of these properties.

An important variant of HoTT is the cubical type theory (CuTT) \cite{cohen_cubical_2018}. Compared with HoTT, CuTT introduces interval variables, allowing users to describe and manipulate paths and shapes explicitly. CuTT also introduces the higher inductive type (HIT), allowing an inductive data type to include equality information. For example, in terms of HIT, the integer type may be defined as
\begin{code}
data Int : Type where
    minus : ℕ → ℕ → Int
    eq : ∀ a b c d → a + d ≡ b + c → minus a b ≡ minus c d
\end{code}
which encodes the mathematical definition $(\N\times\N)/((a,b)\sim(c,d):a+d=b+c)$. Users are required to demonstrate that the equality is preserved when pattern-matching on a HIT. CuTT has multiple implementations in proof assistants, and one of them is Cubical Agda \cite{vezzosi_cubical_2019}. The codebase for this thesis provides the formalisation of some major theorems in Cubical Agda. See Appendix \ref{apdx-agda} for details.

Directed HoTT is another variation of HoTT used to study the synthetic theory of $\infty$-categories \cite{riehl_type_2023}. In directed HoTT, paths may not be invertible, while other properties including function congruence and function extensionality still hold. Table \ref{tab-rosetta} provides a correspondence between category theory, type theory, and homotopy theory.

\section{Conventions of Notations}
The thesis uses a mixed style of set-theoretic and type-theoretic notations. For a term $x$ of type $A$, both the notations $x\in A$ and $x:A$ may occur in the text. The thesis also uses a set-theoretic style notation of subsets, that is
$$\set{x\in A:p(x)}:=\sum_{x:A}p(x)$$
where $p:A\to\U$.

The universe of all types is denoted $\U$, and the universe of all propositions is denoted $\Omega$. We assume that all types in $\U$ are h-sets and all propositions n $\Omega$ are h-props which means that we consider two propositions equal if they are logical equivalent. This assumption is only for the sake of simplicity and can be removed with only minor modifications to the theory\footnote{In fact, we only need to additionally assume that $\I$ is a set.}. The initial and final objects are denoted $\bot$ and $\top$, which are also seen as the falsity and truth in $\Omega$. The only element of $\top$ is $*$.

%% file: chapters/3-interval.tex

\chapter{The Interval Type and Simplices}
\label{chapter-interval}

The interval type plays a pivotal role in SDT in two aspects: It illustrates a primitive homotopy type, as its name suggests, a closed interval, which allows us to define higher-dimensional structures like cubes and simplices, and to talk about homotopy theory in this framework. Meanwhile, it classifies certain subobjects, namely, the open sets on which the partial functions are studied in SDT. This chapter describes how these two aspects of the interval type interact on the properties of simplices, which are important when illustrating information order in homotopy.

\section{The Interval Type}

We denote the interval type as $\I$, along with its two ends as $0$ and $1$, respectively. The existence of the interval type is postulated in the theory. To rule out trivial models, we require $0\ne1$. For any $i\in\I$, denote $\sem{i}:=(i=1)$. To make $\I$ a classifier, we expect $\I$ to be a subuniverse of the universe of propositions $\Omega$, which leads to the following axiom \footnote{Axioms about the interval type will always be postulated throughout the thesis unless otherwise indicated.}:
\begin{axiom}
    \label{ax-prop}
    $\sem{-}:\I\to\Omega$ is an embedding that sends $1\mapsto\top$ and $0\mapsto\bot$.
\end{axiom}

We also observe that $(\Omega,\wedge,\vee,\to,\bot,\top)$ is a bounded distributive lattice, where $\to$ serves as the partial order induced by the lattice. We further expect this embedding to be an embedding of lattice, and as a result, the interval type seen as a subuniverse of propositions will be closed under conjunction and disjunction. 
\begin{axiom}
    \label{ax-lattice}
    $\I$ has a bounded distributive lattice $(\I,\sqcap,\sqcup,\sqsubseteq,0,1)$ structure satisfying
    $$\begin{aligned}
        \sem{i\sqcap j}&=\sem{i}\wedge\sem{j} \\
        \sem{i\sqcup j}&=\sem{i}\vee\sem{j} \\
        i\sqsubseteq j&=\sem{i}\to\sem{j} \\
    \end{aligned}$$
\end{axiom}
It should be emphasised that the name of the interval type does not suggest that the partial order defined here is necessarily a linear order.

\section{The Partial Map Classifier}
We now illustrate how the interval type classifies the subsets and, in addition, the partial maps.
\begin{definition}
    For any $A$, the function type $A\to\I$ is called the \textbf{observational algebra}\footnote{The naming of this concept is attributed to \cite{sterling_domains_2025}.} of $A$, denoted $\Op(A)$. Each element of $\chi:\Op(A)$ identifies a subset $U$ of $A$ defined as
    $$U:=\set{a\in A:\chi(a)=1}=\sum_{a:A}\sem{\chi(a)}$$
    Such a subset is called an \textbf{open subset} of $A$. For simplicity, we may use $U:\Op(A)$ or $U\in\Op(A)$ to refer to the subset itself. The observational algebra is made a bounded distributive lattice by taking pointwise joins and meets.
\end{definition}
One may immediately see that the lattice of the interval induces a bounded distributive lattice on any observational algebra by pointwise operations, which further guarantees that the open sets are closed under finite unions and intersections and that the empty and universal sets are always open. This observation is similar to the definition of the frame of open sets in topology, despite the fact that we so far know nothing about whether open sets are still closed under arbitrary unions. Nevertheless, we still see the observational algebra $\Op(A)$ as the topology of $A$.

With the open set classifier, we can now classify the partial maps defined in the open sets, which is to find a functor $L$ such that $A\to LB$ gives exactly the type of partial maps from $A$ to $B$.

\begin{definition}
    The \textbf{lifting functor} $L$ is defined as
    $$LA:=\sum_{i:\I}(\sem{i}\to A)$$
    that acts on the function $f:A\to B$
    $$Lf:=\lambda(i,x).(i,\lambda\phi.f(x(\phi))):LA\to LB$$
\end{definition}
\begin{theorem}
    \label{thm-classifier}
    The lifting functor classifies partial maps in open subsets. Specifically, for any $A$ and $B$, the set of all partial functions from $A$ to $B$ is isomorphic to $A\to LB$.
\end{theorem}
\begin{proof}
    $$\begin{aligned}
        \sum_{U:\Op(A)}(U\to B)&=\sum_{\chi:\Op(A)}\sum_{a:A}\sem{\chi(a)}\to B \\
        &\cong\sum_{\chi:\Op(A)}\prod_{a:A}\sem{\chi(a)}\to B \\
        &\cong\prod_{a:A}\sum_{\chi:A\to\I}\sem{\chi(a)}\to B \\
        &\cong\prod_{a:A}\sum_{i:\I}\sem{i}\to B \\
        &=A\to LB
    \end{aligned}$$
    Hence, $A\to LB$ classifies the partial maps between $A$ and $B$.
\end{proof}
The lifting functor can be seen as a synthetic counterpart to the maybe monad in the sense that every $LA$ has an element $\bot_A:=(0,\elim_\bot)$ similar to \verb|Nothing| and an embedding
$$\eta_A:A\hookrightarrow LA:=\lambda x.(1,\lambda\_.x)$$
similar to \verb|Just|. In fact, one may consider an model of SDT where $\I$ is just the Boolean type, and the lifting functor in this model is exactly the maybe monad.

\section{The Simplices}
\begin{definition}
    The \textbf{$n$-simplex} $\Delta^n$ in SDT is defined as
    $$\Delta^n:=\set{(i_0,\dots,i_{n-1})\in\I^n:i_0\sqsupseteq\dots\sqsupseteq i_{n-1}}$$
    Specially, $\Delta^0:=\top$ and $\Delta^{-1}:=\bot$.
\end{definition}
Sorting the interval elements in descending order seems counterintuitive, since doing the other way is more common and will illustrate exactly the same structure. The reason we do so will be revealed later, but we still use $\Delta_n$ to denote the $n$-simplex with reversed permutation (i.e. in ascending order) of the interval elements.

One surprising result in SDT shows the connection between the simplices and the lifting functor. One can notice this connection from some simple observations.
\begin{lemma}
    \label{lemma-l-delta-0}
    $L\bot\cong\top$.
\end{lemma}
\begin{proof}
    For any $(i,x)\in L\bot$, $x:\sem{i}\to\bot=\sem{i}\to\sem{0}$, which means $i\sqsubseteq0$. Meanwhile, $0$ is the minimum in $\I$, so we must have $i=0$, and therefore $(0,\rec_\bot)$ is the only element in $L\bot$.
\end{proof}
\begin{lemma}
    \label{lemma-l-delta-1}
    $L\top\cong\I$.
\end{lemma}
\begin{proof}
    Denote the only element of $\top$ by $*$. The isomorphism is given as
    $$\begin{aligned}
        (i,x)&\mapsto i \\
        i&\mapsto(i,\lambda\_.*)
    \end{aligned}$$
    which can be easily verified.
\end{proof}
The observations above give rise to the conjecture that $\Delta^n\cong L^n\top\cong L^{n+1}\bot$. However, difficulties arise when we try to show, for example, $L\I\cong\Delta^2$: We have no measures to construct a tuple of two interval elements from an instance of $(i,x)\in L\I$ since we cannot apply $x$ without knowing anything about $i$. It turns out that we need one more axiom to reach the proof.

\begin{definition}
    \label{def-sigma}
    For any type $A$, an \textbf{$L$-algebra} on $A$ is a function $LA\to A$, and a \textbf{$L$-coalgebra} on $A$ is a function $A\to LA$.
\end{definition}
\begin{axiom}
    \label{ax-l-alg}
    There exists a $L$-algebra $\sigma$ on $\I$ satisfying
    $$\sem{\sigma(i,j)}=\sum_{\phi:\sem{i}}\sem{j(\phi)}$$
    which is called the \textbf{$\sigma$-structure} on $\I$.
\end{axiom}
This axiom suggests that $\I$ as a subuniverse of propositions is closed under $\Sigma$-types. A type like $\I$ that inherit the algebra of propositions via an embedding is called a dominance \cite{reus_general_1997}. It is often additionally required that $\sem{i}\leftrightarrow\sem{j}\implies i=j$, yet this is redundant if we assume $\Omega$ to be the universe of h-props. The intuitive interpretation of the $\sigma$-structure leads to some useful lemmata.
\begin{lemma}
    \label{lemma-sigma-le}
    For any $(i,j)\in L\I$, $\sigma(i,j)\sqsubseteq i$.
\end{lemma}
\begin{proof}
    By Axiom \ref{ax-l-alg}, the left projection $\pi_0$ of $\sum_{\phi:\sem{i}}\sem{j(\phi)}$ has the type $$\pi_0:\sem{\sigma(i,j)}\to\sem{i}$$
    Thus by Axiom \ref{ax-lattice}, $\sigma(i,j)\sqsubseteq i$.
\end{proof}
\begin{lemma}
    \label{lemma-sigma-meet}
    For any $i,j\in\I$, $\sigma(i,\lambda\_.j)=i\sqcap j$.
\end{lemma}
\begin{proof}
    By Axiom \ref{ax-lattice} and \ref{ax-l-alg}, we have
    $$\sem{\sigma(i,\lambda\_.j)}=\sum_{\phi:\sem{i}}\sem{j}=\sem{i}\wedge\sem{j}=\sem{i\sqcap j}$$
    Since $\sem{-}$ is monic by Axiom \ref{ax-prop}, $\sigma(i,\lambda\_.j)=i\sqcap j$.
\end{proof}
\begin{lemma}
    \label{lemma-sigma-eta}
    For any $(i,j)\in L\I$ and $\phi\in\sem{i}$, $j(\phi)=\sigma(i,j)$.
\end{lemma}
\begin{proof}
    Whenever the proposition $\sem{i}\in\Omega$ is inhabited, we can always replace $i$ with $1$ and any $\phi\in\sem{i}$ with $*\in\top$. Thus, we have
    $$\sigma(i,j)=\sigma(1,j)=\sigma(1,\lambda\_.j(*))=1\sqcap j(*)=j(*)=j(\phi)$$
\end{proof}

Although the $\sigma$-structure is powerful enough to prove $L\I\cong\Delta^2$, we need to generalise it to higher dimensions to prove our conjecture in any dimension.
\begin{definition}
    \label{def-sigma-n}
    The $\sigma$-structure on $\I$ can be generalised to $L$-algebras on any $n$-cube $\I^n$ given by
    $$\sigma_n(i,j:\sem{i}\to\I^n):=(\sigma(i,\pi_0\circ j),\dots,\sigma(i,\pi_{n-1}\circ j))$$
    which is called the \textbf{$\sigma$-structure} on $\I^n$.
\end{definition}
Likewise, we have the higher-dimensional counterparts for Lemma \ref{lemma-sigma-le}, \ref{lemma-sigma-meet}, and \ref{lemma-sigma-eta}.
\begin{lemma}
    \label{lemma-sigma-le-n}
    For any $(i,j)\in L\I^n$ and any $0\le k< n$, we have
    $$\pi_k(\sigma_n(i,j))\sqsubseteq i$$
\end{lemma}
\begin{lemma}
    \label{lemma-sigma-meet-n}
    For any $i\in\I$ and $j=(j_0,\dots,j_{n-1})\in\I^n$, we have
    $$\sigma(i,\lambda\_.j)=(i\sqcap j_0,\dots,i\sqcap j_{n-1})$$
\end{lemma}
\begin{lemma}
    \label{lemma-sigma-eta-n}
    For any $(i,j)\in L\I^n$ and $\phi\in\sem{i}$, $j(\phi)=\sigma_n(i,j)$.
\end{lemma}
All the above lemmata are direct consequences of Definition \ref{def-sigma-n} and the corresponding $1$-dimensional cases.

The $\sigma$-structures above are defined on cubes. Before using them to prove our conjecture, we need to restrict it onto simplices.
\begin{lemma}
    The $\sigma$-structure on $\I^n$ defines a $L$-algebra on $\Delta^n\subseteq\I^n$ by restriction.
\end{lemma}
\begin{proof}
    Let $(i,j)\in L\Delta^n$ where $n\ge2$. For any $0\le k\le n-2$,
    $$\forall(\phi:\sem{i}).\pi_k(j(\phi))\sqsupseteq\pi_{k+1}(j(\phi))$$
    Meanwhile, by Definition $\ref{def-sigma}$ and $\ref{def-sigma-n}$, we have
    $$\sem{\pi_k(\sigma_n(i,j))}=\sem{\sigma_n(i,\pi_k\circ j)}=\sum_{\phi:\sem{i}}\sem{\pi_k(j(\phi))}$$
    and similarly
    $$\sem{\pi_{k+1}(\sigma_n(i,j))}=\sum_{\phi:\sem{i}}\sem{\pi_{k+1}(j(\phi))}$$
    Axiom \ref{ax-lattice} assures that the function type $\sem{\pi_{k+1}(j(\phi))}\to\sem{\pi_k(j(\phi))}$ is inhabited, and thus so is $\sem{\pi_{k+1}(\sigma_n(i,j))}\to\sem{\pi_k(\sigma_n(i,j))}$. Therefore, $\sigma_n(i,j)$ gives a descending sequence and is an element of $\Delta^n$.
\end{proof}

We can now prove the conjecture that $\Delta^n\cong L^n\top$
\begin{lemma}[\texttt{SemiLattice.L□↓≡□↓}]
    \label{lemma-l-delta-n}
    For any $n\ge1$, $L\Delta^n\cong\Delta^{n+1}$.
\end{lemma}
\begin{proof}
    We give the isomorphism as follows:
    $$\begin{aligned}
        f(i,j)&:=(i,\sigma_n(i,j))&:L\Delta^n\to\Delta^{n+1} \\
        g(i)&:=(i_0,\lambda\_.(i_1,\dots,i_n))&:\Delta^{n+1}\to L\Delta^n \\
    \end{aligned}$$
    By Lemma $\ref{lemma-sigma-le-n}$, the concatenation $(i,\sigma(i,j))$ is descending and is an element of $\Delta^{n+1}$. This ensures that $f$ is well defined.

    We first show $f\circ g=\id$. Take any $(i_0\sqsupseteq\dots\sqsupseteq i_n)\in\Delta^{n+1}$. We can see
    $$\begin{aligned}
        f(g(i_0,\dots,i_n))=&(i_0,\sigma_n(i_0,\lambda\_.(i_1,\dots,i_n))) \\
        =&(i_0,i_0\sqcap i_1,\dots,i_0\sqcap i_n) \\
        =&(i_0,i_1,\dots,i_n) \\
    \end{aligned}$$
    where $i_0\sqcap i_k=i_k$ for any $0\le k\le n$ since $i_0\sqsupseteq i_k$.

    We then show $g\circ f=\id$. Take any $(i,j)\in L\Delta^n$. We can see
    $$g(f(i,j))=(i,\lambda\_.\sigma_n(i,j))$$
    We will show that $j=\lambda\_.\sigma_n(i,j)$. By the extensionality of functions, this is equivalent to showing that $j(\phi)=\sigma_n(i,j)$ for any $\phi\in\sem{i}$, which is exactly Lemma \ref{lemma-sigma-eta-n}.
\end{proof}
Combining Lemma \ref{lemma-l-delta-0}, \ref{lemma-l-delta-1}, and \ref{lemma-l-delta-n} together, we can immediately reach our goal.
\begin{theorem}[\texttt{SemiLattice.Δ≡□↓}]
    \label{thm-l-delta}
    By iteratively composing the maps defined in Lemma \ref{lemma-l-delta-n}, we can get the following isomorphism:
    $$\Delta^n\cong L^n\top\cong L^{n+1}\bot$$
\end{theorem}
This theorem relates the simplices in homotopy theory to the lifting functor in domain theory. It will enable us to give a homotopy interpretation of the initial (final) $L$-(co)algebra in the future. Before that, we will first investigate the topology of simplices using the Phoa principle.

%% file: chapters/4-phoa.tex

\chapter{The Phoa Principle and the Topology of Simplices}
\label{chapter-phoa}

This chapter investigates the topology of the simplices using the Phoa principle and its generalisation in higher dimensions \cite{pugh_when_2025}. The original Phoa principle asserts that any function $p:\I\to\I$ is monotonic and is uniquely identified by its value at two endpoints, namely $p(0)$ and $p(1)$. Pugh and Sterling generalised it to give similar results on functions from higher simplices to the interval \cite{pugh_when_2025}. We will see that these results identify the topology of simplices with a series of objects with different homotopy structures called spines.

\section{The Phoa Principle}
The homotopy in SDT is expected to give the information order for any type. For the interval type, this means that every path in the interval type witnesses the information order between its two endpoints, or more specifically, the paths on $\I$ shall have the following properties:
\begin{enumerate}
    \item Every function $p:\I\to\I$ satisfies $p(0)\sqsubseteq p(1)$, and
    \item For any $i,j\in\I$ such that $i\sqsubseteq j$, there exists a unique function $p:\I\to\I$ satisfying $p(0)=i$ and $p(1)=j$.
\end{enumerate}
The lattice structure of $\I$ as postulated in Axiom \ref{ax-lattice} already subsumes the existence in the second property as
$$p(t):=i\sqcup(t\sqcap j)$$
automatically satisfies $p(0)=i$ and $p(1)=j$ if $i\sqsubseteq j$. Uniqueness is required in the second property since the paths on $\I$ shall illustrate a preorder, that is, a thin category. Combining the uniqueness and the existence, one will see the following axiom:
\begin{axiom}[Interpolation axiom]
    \label{ax-interpole}
    For any function $p:\I\to\I$,
    $$p(i)=p(0)\sqcup(i\sqcap p(1))$$
\end{axiom}
One can interpret this axiom as stating that every function $\I\to\I$ can be linearly interpolated. This interpolation implies a stronger variant of the first aforementioned property:
\begin{lemma}[\texttt{Lattice.SMonotone}]
    \label{lemma-mono}
    Every function $\I\to\I$ is monotonic.
\end{lemma}
Another important result one can immediately obtain from Axiom \ref{ax-interpole} is the Phoa principle.
\begin{theorem}[Phoa principle, \texttt{Lattice.Phoa}]
    \label{thm-phoa}
    The precomposition of inclusion $\set{0,1}\subseteq\I$ gives an isomorphism
    $$\I^\I\cong\Delta_2=\set{(i,j)\in\I^2:i\sqsubseteq j}$$
\end{theorem}

\section{The Higher Phoa Principle}
One may ask if the results can be generalised in larger dimensions\footnote{Though the dimension of a type is not formally defined, one may still see this as referring to subsets of $n$-cubes.}. A simple example will be the generalisation of the monotonicity:
\begin{lemma}[\cite{pugh_when_2025}, \texttt{Lattice.SMonotoneN}]
    \label{lemma-mono-n}
    For any $n\in\N$ and $f:\I^n\to\I$, $f$ is monotonic with respect to the product poset on $\I^n$.
\end{lemma}
\begin{proof}
    For any $i=(i_1,\dots,i_n)$ and $j=(j_1,\dots,j_n)$ such that $i\sqsubseteq j$, $i_k\sqsubseteq j_k$ for any $1\le k\le n$. Therefore,
    $$\begin{aligned}
        f(i)=&f(i_1,i_2,\dots,i_n) \\
        \sqsubseteq&f(j_1,i_2,\dots,i_n) \\
        \sqsubseteq&f(j_1,j_2,\dots,i_n) \\
        \dots& \\
        \sqsubseteq&f(j_1,j_2,\dots,j_n)=f(j) \\
    \end{aligned}$$
\end{proof}

Recall that our goal is to study the topology of any $n$-simplex $\Delta^n$, which means to find more information about functions of type $\Op(\Delta^n):=\Delta^n\to\I$. The lemma above surprisingly leads to a generalised interpolation of these functions. Analogously to the two endpoints of $\I$, the $n$-simplex has $n+1$ vertices:
\begin{definition}
    The $k$-th \textbf{vertex} counting from $0$ of the $n$-simplex $\Delta^n$ is
    $$v_k:=(1^k0^{n-k})=(\overbrace{1,\dots,1}^\text{$k$ copies},\overbrace{0,\dots,0}^\text{$n-k$ copies})$$
\end{definition}
We can see that $v_a\sqsubseteq v_b$ for any $a\le b$. Thus, by Lemma \ref{lemma-mono-n}, $(f(v_0),\dots,f(v_n))$ gives an ascending sequence of $\I$ and hence is an element of $\Delta_{n+1}$. This process of ``sampling'' a function on vertices defines an embedding from the set of functions $\Delta^n\to\I$ to $\Delta_{n+1}$.

Intuitively, we conjecture that $\I^{\Delta^n}\cong\Delta_{n+1}$, which we call the higher Phoa principle. To prove this, it is necessary to find the inverse of the embedding $\I^{\Delta^n}\hookrightarrow\Delta^{n+1}$, which, as given in \cite{pugh_when_2025}, turns out to be
$$(j_0,\dots,j_n)\in\Delta_{n+1}\mapsto\lambda(i_1,\dots,i_n).j_0\sqcup\bigsqcup_{k=1}^ni_k\sqcap j_k$$
\begin{theorem}[Higher Phoa principle \cite{pugh_when_2025}, \texttt{Lattice.SMonotoneN}]
    \label{thm-phoa-n}
    The evaluation on vertices $\set{v_0,\dots,v_n}\subseteq\Delta^n$ gives an isomorphism
    $$\I^{\Delta^n}\cong\Delta_{n+1}$$
\end{theorem}
\begin{proof}
    We need to prove two parts: The interpolation of values on the vertices reconstructs the function, and the value of a function interpolated from the coefficients in $\Delta_{n+1}$ evaluates to these coefficients on vertices.

    For the first part, we prove this by induction. Assume that $\I^{\Delta^{n-1}}\cong\Delta^n$. Take any $f:\Delta^n\to\I$ and let $j_k:=f(v_k)$ for all $0\le k\le n$. For any $i=(i_1,\dots,i_n)\in\Delta^n$, Let $g:\Delta^{n-1}\to\I$ be
    $$g(i'_1,\dots,i'_{n-1}):=f(i_1\sqcup i_1',i_1',\dots,i'_{n-1})$$
    Thus
    $$\begin{aligned}
        j'_0&:=g(v_0)=f(i_1,0,\dots,0) \\
        j'_k&:=g(v_k)=f(1,1,\dots,0)=f(v_{k+1})=j_{k+1} &1\le k\le n-1 
    \end{aligned}$$
    Since $j'_0$ depends only on $i_1$, by Axiom \ref{ax-interpole} we have
    $$j'_0=f(v_0)\sqcup(i_1\sqcap f(v_1))=j_0\sqcup(i_1\sqcap j_1)$$
    Hence, by the induction hypothesis, we have
    $$g(i'_1,\dots,i'_{n-1})=j'_0\sqcup\bigsqcup_{k=1}^{n-1}i'_k\sqcap j'_k=j_0\sqcup(i_1\sqcap j_1)\sqcup\bigsqcup_{k=1}^{n-1}i'_k\sqcap j_{k+1}$$
    Meanwhile, since $i_1\sqsupseteq i_2$, $i_1\sqcup i_2=i_1$, which gives
    $$\begin{aligned}
        f(i_1,i_2,\dots,i_n)=g(i_2,\dots,i_n)=j_0\sqcup(i_1\sqcap j_1)\sqcup\bigsqcup_{k=1}^{n-1}i_{k+1}\sqcap j_{k+1}=j_0\sqcup\bigsqcup_{k=1}^ni_k\sqcap j_k
    \end{aligned}$$
    Therefore the inductive case is proved. The base case can be simply derived from Theorem \ref{thm-phoa}.

    For the second part, we prove this directly. For any $(j_0,\dots,j_n)\in\Delta_{n+1}$, let
    $$f(i_1,\dots,i_n):=j_0\sqcup\bigsqcup_{k=1}^ni_k\sqcap j_k$$
    Take any $0\le l\le n$, we have
    $$\begin{aligned}
        f(v_l)&=j_0\sqcup\prn{\bigsqcup_{k=1}^l1\sqcap j_k}\sqcup\prn{\bigsqcup_{k=l+1}^n0\sqcap j_k} \\
        &=j_0\sqcup\prn{\bigsqcup_{k=1}^lj_k}\sqcup\prn{\bigsqcup_{k=l+1}^n0} \\
        &=\bigsqcup_{k=0}^lj_k=j_l \\
    \end{aligned}$$
    Therefore the second part is proved.
\end{proof}

\section{Duality between Ascending and Descending Sequences}
One may notice that we do not write the higher Phoa principle in the equivalent form $\I^{\Delta^n}\cong\Delta^{n+1}$, unlike the original statement in \cite{pugh_when_2025}. The reason we distinguish the ascending and descending sequences is due to a certain kind of duality between these two representations of simplices. We can illustrate this duality by studying the lattice structure on simplices.
\begin{lemma}
    $\I^n$, $\Delta^n$, and $\Delta_n$ are bounded distributive lattices under the element-wise joins and meets between sequences.
\end{lemma}
With the interpolation property of the higher Phoa principle and the distributivity, we can also find
\begin{lemma}
    Every function $f:\Delta^n\to\I$ preserves the binary joins. That is, for any $i,j\in\Delta^n$, $f(i\sqcup j)=f(i)\sqcup f(j)$.
\end{lemma}
If we see the join as an analogy to vector addition, the above result suggests that every functional on $\Delta^n$ is linear, despite the fact that $f(0)$ may not be $0$. Now we include this additional requirement.
\begin{definition}
    A \textbf{linear functional} on $\Delta^n$ is a function $f:\Delta^n\to\I$ satisfying
    $$f(0):=f(0,\dots,0)=0$$
    The set of linear functionals on $\Delta^n$ is denoted $(\Delta^n)^*$.
\end{definition}
With the interpolation formula, one can immediately see that a functional is linear iff the interpolation coefficient $j_0=0$, which means that the sequence of coefficients $(j_0,\dots,j_n)\in\Delta_{n+1}$ degenerates to a sequence of a lower dimension $(j_1,\dots,j_n)\in\Delta_n$. This observation can be described in a linear-algebraic style.
\begin{definition}
    The \textbf{inner product} on $\I^n$ is a function $(-\cdot-):\I^n\times\I^n\to\I$, defined as
    $$(i_1,\dots,i_n)\cdot(j_1,\dots,j_n):=\bigsqcup_{k=1}^ni_k\sqcap j_k$$
\end{definition}
\begin{theorem}
    The restriction of the inner product on $\Delta^n$ and $\Delta_n$ induces a duality between $\Delta^n$ and $\Delta_n$. That is, the currying of the inner product $(-\cdot-):\Delta_n\to\Delta^n\to\I$ gives an isomorphism $(\Delta^n)^*\cong\Delta_n$.
\end{theorem}

\section{The Spine and the Simplex}
It is obvious to see that an ascending sequence of $\I$ with $n$ elements is just a monotonic function from the set of indices $[n]:=\set{0\le\dots\le n}$ to $\I$. Meanwhile, a discrete set of $n$ points does not represent the poset $([n],\le)$ in our theory, since a function of type $[n]\to\I$ is not automatically monotonic. Thus, we need to find an object in our theory that precisely represents the order on $[n]$. One may argue that $\Delta^n$ can be such an object, as the higher Phoa suggests. However, $\Delta^n$ is not the only such object, and is not the ``minimal'' one. Take $n=2$ as an example. Geometrically, $\Delta^2$ represents a filled triangle with orientations on the edges:
\[\begin{tikzcd}[sep=large]
    1 & 2 \\
    0
    \arrow["{(1,i)}", from=1-1, to=1-2]
    \arrow["{(i,0)}", from=2-1, to=1-1]
    \arrow["{(i,i)}"', from=2-1, to=1-2]
\end{tikzcd}\]
A ``minimal'' object representing $[2]$ only need to illustrate $0\le1$ and $1\le2$, since the order $\sqsubseteq$ on $\I$ is transitive. Such an object, denoted $\Lambda_2$, should look like an amalgamation of two intervals:
\[\begin{tikzcd}[sep=large]
    1 & 2 \\
    0
    \arrow[from=1-1, to=1-2]
    \arrow[from=2-1, to=1-1]
\end{tikzcd}\]
which can also be seen as a pushout
\[\begin{tikzcd}
    0 && 1 && 2 \\
    & {0\le1} && {1\le2} \\
    && {\Lambda_3}
    \arrow[hook, from=1-1, to=2-2]
    \arrow[hook, from=1-3, to=2-2]
    \arrow[hook, from=1-3, to=2-4]
    \arrow[hook, from=1-5, to=2-4]
    \arrow[from=2-2, to=3-3]
    \arrow[from=2-4, to=3-3]
    \arrow["\lrcorner"{anchor=center, pos=0.125, rotate=135}, draw=none, from=3-3, to=1-3]
\end{tikzcd}\]
Where $0\le1$ and $1\le2$ are both copies of $\I$, with singletons $0$, $1$, and $2$ embeds into the two endpoints of each copy. Represented with higher inductive type (HIT) with directed homotopy, that would be
\begin{code}
data Λ₂ : Type where
    0 : Λ₂
    1 : Λ₂
    2 : Λ₂
    0≤1 : 0 ⊑ 1
    1≤2 : 1 ⊑ 2
\end{code}
The above code can be translated into HIT in Cubical Agda as
\begin{code}
data Λ₂ : Type where
    0 : Λ₂
    1 : Λ₂
    2 : Λ₂
    0≤1 : S → Λ₂
    1≤2 : S → Λ₂
    0≤1-s0 : 0≤1 s0 ≡ 0
    0≤1-s1 : 0≤1 s1 ≡ 1
    1≤2-s0 : 1≤2 s0 ≡ 1
    1≤2-s1 : 1≤2 s1 ≡ 2
\end{code}
To avoid collision with the built-in interval type of the Cubical Agda, the interval type $\I$ along with its two endpoints are named \verb|S|, \verb|s1|, and \verb|s2|. One can see that any function $f:\Lambda_3\to\I$ will always have $f(0)\sqsubseteq f(1)\sqsubseteq f(2)$. Thus, $\I^{\Lambda_2}\cong\Delta^3\cong\I^{\Delta^2}$. We now generalise this observation to larger dimensions.
\begin{definition}
    The \textbf{$n$-spine}, denoted $\Lambda_n$, is the colimit of the following diagram
    \[\begin{tikzcd}
        0 && 1 && 2 & \cdots & n \\
        & {0\le1} && {1\le2} && \cdots
        \arrow[hook, from=1-1, to=2-2]
        \arrow[hook, from=1-3, to=2-2]
        \arrow[hook, from=1-3, to=2-4]
        \arrow[hook, from=1-5, to=2-4]
        \arrow[hook, from=1-5, to=2-6]
        \arrow[hook, from=1-7, to=2-6]
    \end{tikzcd}\]
    with every natural number represent a copy of a singleton, every $k\le k+1$ represent a copy of $\I$, and every morphism points to the endpoints of $\I$. $\Lambda_n$ resembles the shape of $n$ connected intervals
    \[\begin{tikzcd}
        0 & 1 & 2 & \cdots & n
        \arrow[from=1-1, to=1-2]
        \arrow[from=1-2, to=1-3]
        \arrow[from=1-3, to=1-4]
        \arrow[from=1-4, to=1-5]
    \end{tikzcd}\]
\end{definition}
\begin{theorem}[\texttt{Lattice.PhoaΛ}]
    \label{thm-spine}
    The evaluation on vertices $\set{0,\dots,n}\subseteq\Lambda_n$ gives an isomorphism
    $$\I^{\Lambda_n}\cong\Delta_{n+1}$$
\end{theorem}
The formalisation of $\Lambda_n$ and the proof of this theorem are included in the Agda code.

Despite the fact that a spine and a simplex are two different spaces in the sense that a simplex seems to have more points, Theorem \ref{thm-spine} still suggests that they share the same topology. We now make this argument concrete.
\begin{definition}
    If the precomposition of a map $f:A\to B$ between the observational algebras $f^*:\Op(B)\to\Op(A)$ is an isomorphism between lattices, $f$ is called a \textbf{sobriomorphism}, and $A$ and $B$ are said to be \textbf{sobriomorphic}.
\end{definition}
\begin{theorem}
    \label{thm-sobrio}
    The isomorphism $\Op(\Lambda_n)\cong\Op(\Delta^n)$ given by the higher Phoa principle is a sobriomorphism between $\Lambda_n$ and $\Delta^n$.
\end{theorem}
\begin{proof}
    From the higher Phoa principle and Theorem \ref{thm-spine} we can see that a function of type $\Delta^n\to\I$ or $\Lambda_n\to\I$ is uniquely decided by its values on the vertices. Since evaluating a function on a set of points preserves the joins and meets, the isomorphism $\Op(\Lambda_n)\cong\Op(\Delta^n)$ is an isomorphism between the lattices.
\end{proof}

%% file: chapters/5-omega.tex

\chapter{The Transfinite Phoa Principle and the universal \texorpdfstring{$L$}{L}-(co)algebra}
\label{chapter-omega}

This chapter gives a concrete construction of the initial $L$-algebra and the final $L$-coalgebra, both of which will play a pivotal role in investigating the general recursion in SDT. The construction is based on the relation between the lifting functor and the simplices revealed in Chapter \ref{chapter-interval}. Furthermore, this construction will lead to a generalisation of the Phoa principle for the initial $L$-algebra, which enables us to compare the topology of the transfinite counterpart of the simplex and the spine.

\section{The Initial \texorpdfstring{$L$}{L}-algebra and the Final \texorpdfstring{$L$}{L}-coalgebra}
We introduce the concept of the initial $L$-algebra and the final $L$-coalgebra in the language of category theory.
\begin{definition}
    Let $\alpha:LA\to A$ and $\beta:LB\to B$ be two $L$-algebras. A function $f:A\to B$ is a \textbf{homomorphism between the $L$-algebras} $\alpha$ and $\beta$ iff the following diagram commutes:
    \[\begin{tikzcd}
        LA & LB \\
        A & B
        \arrow["Lf", from=1-1, to=1-2]
        \arrow["\alpha"', from=1-1, to=2-1]
        \arrow["\beta", from=1-2, to=2-2]
        \arrow["f"', from=2-1, to=2-2]
    \end{tikzcd}\]
    The \textbf{category of $L$-algebras} is a category where objects are $L$-algebras and the morphisms are homomorphisms between the $L$-algebras. Dually, we also have the \textbf{category of $L$-coalgebras}.
\end{definition}
\begin{definition}
    The \textbf{initial $L$-algebra} is the initial object of the category of $L$-algebras, denoted $\phi:L\w\to\w$. The \textbf{final $L$-coalgebra} is the final object of the category of $L$-coalgebras, denoted $\psi:\wb\to L\wb$.
\end{definition}

One may attempt to construct the initial $L$-algebra and the final $L$-coalgebra using the Ad\'amek's fixed point theorem \cite{adamek_free_1974}, which asserts that the initial $L$-algebra is the colimit of the following diagram:
\[\begin{tikzcd}
    \bot & {L\bot} & {L^2\bot} & {L^3\bot} & \cdots
    \arrow["{!}", from=1-1, to=1-2]
    \arrow["{L!}", from=1-2, to=1-3]
    \arrow["{L^2!}", from=1-3, to=1-4]
    \arrow[from=1-4, to=1-5]
\end{tikzcd}\]
Unfortunately, this construction only works if $L$ preserves the $\omega$-chains, or equivalently, the initial $L$-algebra is inductive, which was once believed to be true \cite{hyland_first_1991}, but then disproved with counterexamples in the effective topos \cite{van_oosten_axioms_2000}. To proceed with our investigation, we treat this as an axiom.
\begin{axiom}
    \label{ax-inductive}
    The initial $L$-algebra $\omega$ exists and satisfies $\omega=\varinjlim_{n\in\N}L^n\bot$.
\end{axiom}

It is difficult to illustrate this colimit directly. However, Theorem \ref{thm-l-delta} gives rise to an isomorphic chain built from simplices:
\[\begin{tikzcd}
    \bot & {L\bot} & {L^2\bot} & {L^3\bot} & \cdots \\
    {\Delta^{-1}} & {\Delta^0} & {\Delta^1} & {\Delta^2} & \cdots
    \arrow["{!}", from=1-1, to=1-2]
    \arrow["\cong", from=1-1, to=2-1]
    \arrow["{L!}", from=1-2, to=1-3]
    \arrow["\cong", from=1-2, to=2-2]
    \arrow["{L^2!}", from=1-3, to=1-4]
    \arrow["\cong", from=1-3, to=2-3]
    \arrow[from=1-4, to=1-5]
    \arrow["\cong", from=1-4, to=2-4]
    \arrow["{s_0}", from=2-1, to=2-2]
    \arrow["{s_1}", from=2-2, to=2-3]
    \arrow["{s_2}", from=2-3, to=2-4]
    \arrow[from=2-4, to=2-5]
\end{tikzcd}\]
where $!=\rec_\bot$. We only need to find the exact form of $s_n:\Delta^{n-1}\to\Delta^n$. Recall that $L\bot=\set{(0,\rec_\bot)}$ and $Lf=\lambda(i,x).(i,f\circ x)$. Thus,
$$L!(0,\rec_\bot)=(0,\rec_\bot\circ\rec_\bot)=(0,(0,\rec_\bot))$$
which corresponds to $0\in\I=\Delta^1$. This means that $s_1(*)=0$ and more generally
$$s_n(i_1,\dots,i_{n-1})=(i_1,\dots,i_{n-1},0)$$
Therefore,
$$\Delta^\omega:=\varinjlim_{n\in\N}s_n=\prn{\sum_{n\in\N}\Delta^n}/(x\sim s_n(x))$$
We then further investigate what the elements of $\Delta^\omega$ exactly are. From the definition, $\Delta^\omega$ should contain all descending sequences of $\I$ of any finite length, up to injections by $s_n$. Let
$$\Delta^\infty:=\set{i\in\I^\N:\forall n\in\N.i(n)\sqsupseteq i(n+1)}$$
We can see that any $\Delta^n$ embeds into $\Delta^\infty$ via
$$\iota_n(i_1,\dots,i_n):=(i_1,\dots,i_n,0,\dots)$$
and that embedding commutes to $s_n$. Hence, $\Delta^\omega$ is exactly the subset of $\Delta^\infty$ containing the images of the embeddings. That is,
$$\Delta^\omega=\set{i\in\Delta^\infty:\exists n\in\N.i(n)=0}$$

We now construct the final $L$-coalgebra, which can be done by calculating the limit of the diagram
\[\begin{tikzcd}
    \top & {L\top} & {L^2\top} & {L^3\top} & \cdots \\
    {\Delta^0} & {\Delta^1} & {\Delta^2} & {\Delta^3} & \cdots
    \arrow["\cong", from=1-1, to=2-1]
    \arrow["{*}"', from=1-2, to=1-1]
    \arrow["\cong", from=1-2, to=2-2]
    \arrow["{L*}"', from=1-3, to=1-2]
    \arrow["\cong", from=1-3, to=2-3]
    \arrow["{L^2*}"', from=1-4, to=1-3]
    \arrow["\cong", from=1-4, to=2-4]
    \arrow[from=1-5, to=1-4]
    \arrow["{d_0}"', from=2-2, to=2-1]
    \arrow["{d_1}"', from=2-3, to=2-2]
    \arrow["{d_2}"', from=2-4, to=2-3]
    \arrow[from=2-5, to=2-4]
\end{tikzcd}\]
Similarly to the previous observations, it can be seen that $d_1(i,j)=(i)$, and more generally,
$$d_n(i_1,\dots,i_n)=(i_1,\dots,i_{n-1})$$
One can immediately realise that
$$\varprojlim_{n\in\N}d_n=\Delta^\infty$$
\begin{definition}
    $\Delta^\omega$ is the $\omega$-simplex and $\Delta^\infty$ is the $\infty$-simplex.
\end{definition}
\begin{theorem}
    $\Delta^\omega\cong\w$ carries the initial $L$-algebra
    $$\phi(i,j)=(i,\sigma_\omega(i,j))$$
    where
    $$\sigma_\omega(i,j:\sem{i}\to\I^\N):=\lambda n.\sigma(i,j(-)(n))$$

    $\Delta^\infty\cong\wb$ carries the final $L$-coalgebra
    $$\psi(i)=(i(0),\lambda\_.\lambda n.i(n+1))$$
\end{theorem}

Similarly to $\Delta_n$, we also have
$$\Delta_\infty:=\set{i\in\I^\N:\forall n\in\N.i(n)\sqsubseteq i(n+1)}$$
and
$$\Delta_\omega=\set{i\in\Delta_\infty:\exists n\in\N.i(n)=1}$$
Though these objects do not occur as an initial(final) $L$-coalgebra, we will find their connection with $\Delta^\omega$ and $\Delta^\infty$ via a transfinite variant of the Phoa principle.

\section{The Transfinite Phoa Principle}
Recall that the higher Phoa principle samples a function $\Delta^n\to\I$ on the $n+1$ vertices. We now repeat this sampling process on $\Delta^\omega$. It is not difficult to see that the vertices mapped from $\Delta^n$ to $\Delta^\omega$ via $\iota_n$ are
$$v_k:=(1^k\overline{0})=(\overbrace{1,\dots,1}^\text{$k$ copies},0,\dots)$$
which defines an embedding $\N\hookrightarrow\Delta^\omega$.
\begin{lemma}
    For any function $f:\Delta^\omega\to\I$, the evaluation
    $$j:=\lambda n.f(v_n)=f(1^n\overline{0})$$
    gives an increasing sequence, and hence $j\in\Delta_\infty$.
\end{lemma}
\begin{proof}
    For any $n\ge1$, $v_{n-1},v_n\in\Delta^\omega$ can be seen as embedded from the corresponding vertices in the $n$-simplex. By precomposing the embedding $\Delta^n\hookrightarrow\Delta^\omega$ with $f$, we get its restriction on $\Delta^n$, which is monotonic by Lemma \ref{lemma-mono-n}. Therefore, $f(v_{n-1})\sqsubseteq f(n)$.
\end{proof}

This observation is enough for us to conjecture that $\I^{\Delta^\omega}\cong\Delta_\infty$, which we call the transfinite Phoa principle. Since both $\Delta^\omega$ and $\Delta_\infty$ are constructed by (co)limit, we can prove the transfinite Phoa principle from the higher version.
\begin{lemma}
    \label{lemma-chain}
    The higher Phoa principle gives an isomorphism between chains
    \[\begin{tikzcd}
        {\mathbb{I}^{\Delta^{-1}}} & {\mathbb{I}^{\Delta^0}} & {\mathbb{I}^{\Delta^1}} & {\mathbb{I}^{\Delta^2}} & \cdots \\
        {\Delta_0} & {\Delta_1} & {\Delta_2} & {\Delta_3} & \cdots
        \arrow["\cong", from=1-1, to=2-1]
        \arrow["{\mathbb{I}^{s_0}}"', from=1-2, to=1-1]
        \arrow["\cong", from=1-2, to=2-2]
        \arrow["{\mathbb{I}^{s_1}}"', from=1-3, to=1-2]
        \arrow["\cong", from=1-3, to=2-3]
        \arrow["{\mathbb{I}^{s_2}}"', from=1-4, to=1-3]
        \arrow["\cong", from=1-4, to=2-4]
        \arrow[from=1-5, to=1-4]
        \arrow["{d^0}"', from=2-2, to=2-1]
        \arrow["{d^1}"', from=2-3, to=2-2]
        \arrow["{d^2}"', from=2-4, to=2-3]
        \arrow[from=2-5, to=2-4]
    \end{tikzcd}\]
    where
    $$d^n:\Delta_{n+1}\to\Delta_n:=\lambda(i_1,\dots,i_{n+1}).(i_1,\dots,i_n)$$
\end{lemma}
\begin{proof}
    Let $f:\Delta^n\to\I$ and $j=(j_0,\dots,j_n):=(f(v_0),\dots,f(v_n))\in\Delta_{n+1}$. For any $(i_1,\dots,i_{n-1})\in\Delta^{n-1}$, we have
    $$\begin{aligned}
        &\I^{s_n}(f)(i_1,\dots,i_{n-1}) \\
        =&(f\circ s_n)(i_1,\dots,i_{n-1}) \\
        =&f(s_n(i_1,\dots,i_{n-1})) \\
        =&f(i_1,\dots,i_{n-1},0) \\
    \end{aligned}$$
    and hence $\I^{s_n}(f)(v_k)=f(v_k)$ for $v_0,\dots,v_{n-1}\in\Delta^{n-1}$. Thus,
    $$(\I^{s_n}(f)(v_0),\dots,\I^{s_n}(f)(v_{n-1}))=(j_0,\dots,j_{n-1})=d^n(j)$$
    Therefore, the diagram between the chains commutes, and the higher Phoa principle gives an isomorphism between them.
\end{proof}
\begin{theorem}[Transfinite Phoa principle, \texttt{Omega.Phoaω}]
    \label{thm-phoa-omega}
    The evaluation on vertices $v:\N\hookrightarrow\Delta^\omega$ gives an isomorphism
    $$\I^{\Delta^\omega}\cong\Delta_\infty$$
\end{theorem}
\begin{proof}
    From Lemma \ref{lemma-chain} and the property of the internal-hom, we have
    $$\I^{\Delta^\omega}=[\I,\Delta^\omega]=[\I,\varinjlim_{n\in\N}s_n]=\varprojlim_{n\in\N}[\I,s_n]=\varprojlim_{n\in\N}\I^{s_n}\cong\varprojlim_{n\in\N}d^n=\Delta_\infty$$
\end{proof}
One can realise that this proof by (co)limits only works for $\I^{\Delta^n}\cong\Delta_{n+1}$, and fails to proceed with $\I^{\Delta^n}\cong\Delta^{n+1}$ since the diagram no longer commutes if we replace $d^n$ with $d_n$. This makes a stronger argument for why we should treat $\Delta^n$ and $\Delta_n$ as dual but different objects, even though we obviously have $\Delta^n\cong\Delta_n$ for the finite case.

\section{The \texorpdfstring{$\omega$}{omega}-spine}
Similarly to the simplices, there also exists a transfinite counterpart for the spines, denoted $\Lambda_\omega$. Recall that a $n$-spine $\Lambda_n$ contains the minimal information to describe the order on $[n]=\set{0,\dots,n}$. A transfinite version of the spine should likewise grasp the information that generates the order on $\N$, which we can see is the set of relations $\set{n\le n+1:n\in\N}$. By extending the diagram used to define $\Lambda_n$, we can define $\Lambda_\omega$ as the limit of the following diagram:
\[\begin{tikzcd}
    0 && 1 && 2 & \cdots \\
    & {0\le1} && {1\le2} && \cdots
    \arrow[hook, from=1-1, to=2-2]
    \arrow[hook, from=1-3, to=2-2]
    \arrow[hook, from=1-3, to=2-4]
    \arrow[hook, from=1-5, to=2-4]
    \arrow[hook, from=1-5, to=2-6]
\end{tikzcd}\]
Though this diagram is infinite, we can find an equivalent finite limit by indexing the objects in the diagram using $\N$.
\begin{definition}
    The \textbf{$\omega$-spine} $\Lambda_\omega$ is the limit of the following diagram:
    \[\begin{tikzcd}
        \N \\
        {\N\times\I}
        \arrow["{p_0}"', curve={height=6pt}, from=1-1, to=2-1]
        \arrow["{p_1}", curve={height=-6pt}, from=1-1, to=2-1]
    \end{tikzcd}\]
    where $p_0(n):=(n,0)$ and $p_1(n):=(n+1,1)$. Equivalently, $\Lambda_\omega$ is the coequalizer of $p_0$ and $p_1$.
\end{definition}
The $\omega$-spine described as an HIT in directed homotopy is
\begin{code}
data Λω : Type where
    step : ℕ → Λω
    n≤n+1 : ∀ n → step n ⊑ step (suc n)
\end{code}
which can be translated into an HIT in Cubical Agda as
\begin{code}
data Λω : Type where
    step : ℕ → Λω
    n≤n+1 : ℕ → S → Λω
    n≤n+1-s0 : ∀ n → n≤n+1 n s0 ≡ step n
    n≤n+1-s1 : ∀ n → n≤n+1 n s1 ≡ step (suc n)
\end{code}
We can then get the following results:
\begin{theorem}[\texttt{Omega.PhoaΛω}]
    \label{thm-spine-omega}
    The evaluation on vertices $v:\N\hookrightarrow\Lambda_\omega$ gives an isomorphism
    $$\I^{\Lambda_\omega}\cong\Delta_\infty$$
\end{theorem}
\begin{theorem}
    \label{thm-sobrio-omega}
    $\Lambda_\omega$ and $\Delta^\omega$ are sobriomorphic via the isomorphism given by the transfinite Phoa principle.
\end{theorem}

%% file: chapters/6-complete.tex

\chapter{Completeness and Topology in Synthetic Domain Theory}
\label{chapter-topology}


\section{Orthogonality}
The Phoa principle and its variants are all statements of the form that ``a function can be uniquely identified by its values on a subset of its domain''. Expressed in commutative diagrams, that is
\[\begin{tikzcd}
    X & Y \\
    Z
    \arrow["f", from=1-1, to=1-2]
    \arrow["h"', from=1-1, to=2-1]
    \arrow["{\overline{h}}", dashed, from=1-2, to=2-1]
\end{tikzcd}\]
where $h$ uniquely factors through $f$ as $h=\overline{h}\circ f$.
Reus and Streicher gave a generalisation of statements of such kind in \cite{reus_general_1997}.
\begin{definition}
    For $f:X\to Y$ and $g:Z\to W$, if for any $h:X\to Z$ and $k:Y\to U$, there is a unique $\alpha:Y\to Z$ such that the following diagram commutes:
    \[\begin{tikzcd}
        X & Y \\
        Z & U
        \arrow["f", from=1-1, to=1-2]
        \arrow["h"', from=1-1, to=2-1]
        \arrow["\alpha"{description}, dashed, from=1-2, to=2-1]
        \arrow["k", from=1-2, to=2-2]
        \arrow["g"', from=2-1, to=2-2]
    \end{tikzcd}\]
    then $f$ is said to be \textbf{orthogonal} to $g$, denoted $f\perp g$. If $U=\top$ is the final object, $f$ is said to be \textbf{orthogonal} to $Z$, denoted $f\perp Z$, and $Z$ is said to be \textbf{right-orthogonal} to $f$, or \textbf{$f$-orthogonal}.
\end{definition}
An important result on completeness without any prior knowledge of $f$ is provided in \cite{reus_general_1997}.
\begin{theorem}
    \label{thm-prod}
    If for any $a\in A$, $B(a)$ is $f$-orthogonal, then $\prod_{a:A}B(a)$ is $f$-orthogonal.
\end{theorem}

There are several completeness properties involved in the foundations of domains in SDT. This chapter investigates the conditions for the interval type and, more generally, the observational algebras to satisfy those properties. This chapter will also attempt to illustrate the connections between the concept of sobriomorphisms and these properties.

\section{Orders from Paths and Completeness Properties}
\label{sect-order}

The purpose of introducing the interval type in SDT is to identify the information order on a type $X$ with the paths on it. We now see what properties we need to make the paths a domain.
\begin{definition}
    Let $A$ be a type and $x,y\in A$. The \textbf{path type} between $x$ and $y$, written as $\hom_A(x,y)$, $\hom(x,y)$, or $x\leadsto y$ if without ambiguity, is defined as
    $$\hom_A(x,y):=\set{p\in A^\I:p(0)=x\wedge p(1)=y}$$
\end{definition}
\begin{lemma}
    For any type $A$, $(-\leadsto-):A\to A\to \U$ is reflexive.
\end{lemma}
\begin{proof}
    $$\id:=\lambda x.\lambda\_.x:\forall(x:A).x\leadsto x$$
\end{proof}
\begin{lemma}
    Take any type $A$, $B$. $(-\leadsto-)$ is congruent to all functions $f:A\to B$.
\end{lemma}
\begin{proof}
    $$\lambda x.\lambda y.\lambda(p:x\leadsto y).f\circ p:\forall(x:A).\forall(y:A).(x\leadsto y)\to(f(x)\leadsto f(y))$$
\end{proof}
\begin{lemma}
    Take any type $A$, $B$, and two functions $f,g:A\to B$. If $f(x)\leadsto g(x)$, $f\leadsto g$.
\end{lemma}
\begin{proof}
    $$\lambda(P:\forall x.f(x)\leadsto g(x)).\lambda i.\lambda x.P(x)(i):(\forall x.f(x)\leadsto g(x))\to(f\leadsto g)$$
\end{proof}

So far we can see that the path type illustrates a reflexive relation that is persevered by all functions and has function extensionality on function types. We now introduce the completeness properties that force a type to be a synthetic category and, moreover, a synthetic poset. The first completeness property is the Segal-completeness that illustrates the composition of paths.
\begin{definition}
    A type $A$ is called \textbf{Segal-complete} (or \textbf{Segal} for short) if it is right-orthogonal to $\Lambda_2\hookrightarrow\Delta^2$. That is, every function of type $\Lambda_2\to A$ extends uniquely to a function of type $\Delta^2\to A$.
\end{definition}
In a Segal type, paths can be composed as shown in the following diagram:
\[\begin{tikzcd}
    {1\mapsto y} & {2\mapsto z} \\
    {0\mapsto x}
    \arrow["q", squiggly, from=1-1, to=1-2]
    \arrow["p", squiggly, from=2-1, to=1-1]
    \arrow["{p\cdot q}"', squiggly, from=2-1, to=1-2]
\end{tikzcd}\]

Next, we introduce the Rezk-completeness following the definitions in \cite{buchholtz_synthetic_2022,sterling_domains_2025}.
\begin{definition}
    Let $\E$ be the colimit of the following diagram:
    \[\begin{tikzcd}
        & \I && \I && \I \\
        \top && {\Delta^2} && {\Delta^2} && \top
        \arrow[from=1-2, to=2-1]
        \arrow["{\lambda i.(i,i)}"{description}, from=1-2, to=2-3]
        \arrow["{\lambda i.(i,0)}"{description}, from=1-4, to=2-3]
        \arrow["{\lambda i.(1,i)}"{description}, from=1-4, to=2-5]
        \arrow["{\lambda i.(i,i)}"{description}, from=1-6, to=2-5]
        \arrow[from=1-6, to=2-7]
    \end{tikzcd}\]
    A type $A$ is called \textbf{Rezk-complete} (or \textbf{Rezk} for short) if it is right-orthogonal to $\E\to\top$.
\end{definition}
Functions of type $\E\to A$ classify all isomorphism of the Segal type $A$, and the Rezk-completeness asserts that all isomorphisms are equivalences \cite{buchholtz_synthetic_2022}. If $A$ is both Segal and Rezk, $A$ is said to be a \textbf{synthetic category}.

To make a category a poset, we additionally require that parallel paths to be equivalent, which is described by $\I$-separateness \cite{sterling_domains_2025}.
\begin{definition}
    Let $\I_\parallel$ be the following pushout:
    \[\begin{tikzcd}
        {\set{0,1}} & \I \\
        \I & {\I_\parallel}
        \arrow[hook, from=1-1, to=1-2]
        \arrow[hook, from=1-1, to=2-1]
        \arrow[from=1-2, to=2-2]
        \arrow[from=2-1, to=2-2]
        \arrow["\lrcorner"{anchor=center, pos=0.125, rotate=180}, draw=none, from=2-2, to=1-1]
    \end{tikzcd}\]
    A type $A$ is called \textbf{$\I$-separated} if it is right-orthogonal to $\I_\parallel\to\I$.
\end{definition}
Functions of type $\I_\parallel\to A$ classify all pairs of parallel paths, and $\I$-separateness guarantees that there exists at most one path between any two points, which truncates a category to a poset\footnote{The antisymmetry is guaranteed by Rezk-completeness.}. An $\I$-separated synthetic category is called a \textbf{synthetic poset}.

In SDT, we care about posets with suprema for all $\omega$-chains, which gives an additional completeness property.
\begin{definition}
    A type $A$ is called \textbf{chain-complete} (or \textbf{complete} for short) if it is right-orthogonal to the embedding between the initial and final $L$-(co)algebras $\w\hookrightarrow\wb$.
\end{definition}
A chain-complete synthetic poset is called a \textbf{synthetic predomain}. If it in addition contains a minimum, it is called a \textbf{synthetic domain}.

\section{Topology and Completeness}
We now investigate under what conditions is the interval type a synthetic predomain. Notice that Theorem \ref{thm-prod} implies that an orthogonality property holds for all observational algebras iff it holds for the interval type $\I$. In this section, we only postulate Axiom \ref{ax-prop}, \ref{ax-lattice}, and \ref{ax-l-alg}, as we will relate Axiom \ref{ax-interpole} with completeness properties and other statements.

We first show that Axiom \ref{ax-interpole} is enough to illustrate a poset.
\begin{theorem}
    \label{thm-poset}
    $\I$ is a synthetic poset iff Axiom \ref{ax-interpole} holds.
\end{theorem}
\begin{proof}
    If $\I$ is a synthetic poset, for any $i\sqsubseteq j$, there is a unique path between $i$ and $j$, which implies Axiom \ref{ax-interpole}. In reverse, with Axiom \ref{ax-interpole}, we have $(i\sqsubseteq j)\cong(i\leadsto j)$, which makes $\leadsto$ a partial order on $\I$.
\end{proof}
Axiom \ref{ax-interpole} does not tell us whether $\I$ is chain-complete, so we have to make it an axiom.
\begin{axiom}
    \label{ax-omega}
    $\I$ is chain-complete.
\end{axiom}
\begin{theorem}
    \label{thm-domain}
    $\I$ is a synthetic domain iff Axiom \ref{ax-interpole} and \ref{ax-omega} holds.
\end{theorem}
\begin{proof}
    Derived immediately from Theorem \ref{thm-poset} and the fact that $0$ is the minimum of $\I$.
\end{proof}

Recall that we use the observational algebra $\Op(X):=\I^X$ to describe the topology of a space $X$, yet this structure only provides finite unions and intersections, while the common definition for a topology requires the open sets to be closed under arbitrary unions. Now, with the condition that $\I$ is a synthetic domain, we have countable unions given by
$$\bigsqcup_{n\in\N}\chi_n=\sup_{n\in\N}\bigsqcup_{k=0}^n\chi_k$$
Although this still differs from having arbitrary unions, one can see this as illustrating the topology of a second-countable space.

Now, we provide an alternative way to describe the completeness properties using the concept of sobriomorphisms.
\begin{definition}
    If $S\subseteq X$ and the inclusion $S\hookrightarrow X$ is a sobriomorphism, A is called a \textbf{spanning subset} of $X$, denoted $S\unlhd X$.
\end{definition}
\begin{lemma}
    If $A\unlhd B\unlhd C$, $A\unlhd C$.
\end{lemma}
\begin{proof}
    Derived immediately from the composition of isomorphisms $\Op(A)\cong\Op(B)\cong\Op(C)$.
\end{proof}
\begin{lemma}
    \label{lemma-intermediate}
    If $A\unlhd C$ and $A\subseteq B\subseteq C$, $A\unlhd B\unlhd C$.
\end{lemma}
\begin{proof}
    The restriction $(-|_A):\Op(C)\to\Op(A)$ factors as $\Op(C)\to\Op(B)\to\Op(A)$. Since $(-|_A)$ is invertible, so is $(-|_B):\Op(C)\to\Op(B)$ and $(-|_A):\Op(B)\to\Op(A)$.
\end{proof}
Recall that from Axiom \ref{ax-interpole} we have shown Theorem \ref{thm-sobrio} and \ref{thm-sobrio-omega}, which in reverse can be seen as a condition implying Axiom \ref{ax-interpole} itself.
\begin{lemma}
    $\Lambda_\omega\unlhd\Delta^\omega$ implies that $\I$ is a synthetic poset.
\end{lemma}
We also have an analogous condition for Axiom \ref{ax-omega}.
\begin{lemma}
    If $\I$ is a synthetic poset, $\Delta^\omega\unlhd\Delta^\infty$ implies that $\I$ is a synthetic domain.
\end{lemma}
Now we combine these two conditions using Lemma \ref{lemma-intermediate}.
\begin{theorem}
    \label{thm-last}
    $\Lambda_\omega\unlhd\Delta^\infty$ implies that $\I$ is a synthetic domain.
\end{theorem}
We now propose a conjecture from the observation of Theorem \ref{thm-last} to conclude the investigations on sobriomorphisms so far:
\begin{conjecture}
    \label{hypo-last}
    A type $A$ is a synthetic predomain if it is right-orthogonal to $\Lambda_\omega\hookrightarrow\wb$.
\end{conjecture}
Although this conjecture may not be true, we still wish that the gap between orthogonality and completeness of a synthetic predomain is reachable within a few additional properties. More specifically, we expect that, under certain situations, the orthogonality to the embedding above implies the Segal and chain completeness, since the embedding is the composition of $\Lambda_\omega\hookrightarrow\w$ and $\w\hookrightarrow\wb$.

%% file: chapters/7-conclusion.tex

\chapter{Conclusions}

In this thesis, we have shown a new generalisation of Phoa principle that applies on the initial $L$-algebra. This generalisation suggests a possible connection between the Segal completeness and the chain completeness. We expect that this hidden connection will be made explicit in future and will eventually guide us on constructing the models for synthetic domain theory.

Nevertheless, we should admit that the route which brings us from the classical Phoa principle to the transfinite Phoa principle heavily relies on the lattice structure of the interval type. As a result, Conjecture \ref{hypo-last} may see a non-negligible gap on types other than the observational algebras. One may expect a duality between the ``ordinary'' types and their observational algebras analogous to the Stone duality between the sober spaces and their observational algebras of open sets. One may ask if applying the functor $\Op$ twice gives such duality. It turns out that, if we define the a sober space to be a space $X$ satisfying $X\cong\Op^2(X)$, none of the simplices (either finite or transfinite) are sober, though we have the intuition that simplices contain the most points among all spaces capturing the same order structure. We hereby list two questions that remain to answer:
\begin{enumerate}
    \item Does the observational algebra $\Op(X):=X\to\I$ really behave like a synthetic counterpart of the frame of open sets, even though it does not guarantee arbitrary unions?
    \item Is there a general way to assign each ``good'' type (Segal, Rezk, $\I$-separated, etc.) a lattice that captures the information order and in reverse, reconstruct a type from the observational algebra?
\end{enumerate}
Sterling and Ye proposed another duality between spaces and lattices \cite{sterling_domains_2025}. They see $\Op$ as a functor sending an ``affine space'' to $\I$-algebras, and see the right adjoint of this functor as finding the spectrum of a $\I$-algebra. This framework might lead to a possible way to prove or disprove Conjecture \ref{hypo-last}.

This thesis also provides formalisations of the proofs in Cubical Agda. The formalisation codebase includes the following results:
\begin{enumerate}
    \item Axiomatisation of the lattice and $L$-algebra on $\I$.
    \item Relation between $L$-algebras and simplices.
    \item Classical and higher Phoa principles and their variants on spines.
    \item Transfinite Phoa principle and its variant on $\omega$-spine.
\end{enumerate}
We expect that the explicit constructions in our formalisation could contribute as a foundation for the future development of a full-scale calculus for synthetic domain theory.

Kudasov et al. have delivered a proof-assistant call Rzk \cite{kudasov_formalizing_2023} based on the synthetic $\infty$-category theory by Riehl et al. \cite{riehl_type_2023}, and it is promising to migrate our work into Rzk to utilize its native support for directed homotopy. Another work that provides a formalisation framework for synthetic domain theory is the guarded mode in Agda, based on the theory in \cite{mogelberg_bisimulation_2019} and \cite{kristensen_greatest_2022}. Despite the fact that the guarded mode is designed for the synthetic guarded domain theory \cite{birkedal_first_2012} instead of SDT mentioned in this thesis, we still see the possibility of implementing SDT as a language extension of Agda in the future.


%% file: chapters/a-proof.tex

\chapter{Formalizing Proofs in Agda}
\label{apdx-agda}

We have formalized some of the major theorems and lemmata in (Cubical) Agda. Table \ref{tab-agda} lists these results and their corresponding formalisations in the code base. This appendix will give a walkthrough of the technical details. Notice that the code snippets in the appendix are not necessarily the same as the original code, and some simplifications are made to avoid distracting readers to irrelevant technical details (e.g., universe levels, aliased constructors, etc.).

\begin{table}[h]
    \centering
    \begin{tabular}{l|l|l}
    Theorem or Lemma & Module & Term \\ \hline
    Lemma \ref{lemma-l-delta-n} & \verb|SemiLattice| & \verb|L□↓≡□↓| \\ \hline
    Theorem \ref{thm-l-delta} & \verb|SemiLattice| & \verb|Δ≡□↓| \\ \hline
    Lemma \ref{lemma-mono} & \verb|Lattice| & \verb|SMonotone| \\ \hline
    Theorem \ref{thm-phoa} & \verb|Lattice| & \verb|Phoa| \\ \hline
    Lemma \ref{lemma-mono-n} & \verb|Lattice| & \verb|SMonotoneN| \\ \hline
    Theorem \ref{thm-phoa-n} & \verb|Lattice| & \verb|PhoaN| \\ \hline
    Theorem \ref{thm-spine} & \verb|Lattice| & \verb|PhoaΛ| \\ \hline
    Theorem \ref{thm-phoa-omega} & \verb|Omega| & \verb|Phoaω| \\ \hline
    Theorem \ref{thm-spine-omega} & \verb|Omega| & \verb|PhoaΛω| \\ 
    \end{tabular}
    \caption{Formalized Theorem and Lemmata in Agda}
    \label{tab-agda}
\end{table}

\section{Postulating the Interval Type}
In SDT, the interval type $\I$ is a non-trivial bounded distributive lattice equipped with a $L$-algebra as mentioned in Axiom \ref{ax-l-alg}. Although it is feasible to list all these properties as axioms, we seek a more concise and elegant way to do it ``bottom-up''. That is, we only postulate the minimal information on the algebraic structures on $\I$ and prove other properties when necessary. The selected postulates are:
\begin{enumerate}
    \item (Axiom \ref{ax-prop}, \verb|PreSDT.SisSet|) $\I$ is a h-set.
    \item (Axiom \ref{ax-prop}, \verb|PreSDT.s0≠s1|) $0\in\I$ and $1\in\I$ with $0\ne1$.
    \item (Axiom \ref{ax-prop}, \verb|PreSDT.defIsMono|) $\sem{i}=\sem{j}$ implies $i=j$.
    \item (Axiom \ref{ax-l-alg}, \verb|SemiLattice.SΣ-def|) There exists a $L$-algebra $\sigma:L\I\to\I$ satisfying
    $$\sem{\sigma(i,j)}=\sum_{\phi:\sem{i}}\sem{j(\phi)}$$
    Define $i\sqcap j:=\sigma(i,\lambda\_.j)$.
    \item (Axiom \ref{ax-lattice}, \verb|Lattice.⊔-def|) The exists a binary operation $-\sqcup-:\I\to\I\to\I$ satisfying
    $$\sem{i\sqcup j}=\sem{i}*\sem{j}$$
    where $A*B$ is the pushout
    \[\begin{tikzcd}
        {A\times B} & B \\
        A & {A*B}
        \arrow["{\pi_2}", from=1-1, to=1-2]
        \arrow["{\pi_1}"', from=1-1, to=2-1]
        \arrow[from=1-2, to=2-2]
        \arrow[from=2-1, to=2-2]
        \arrow["\lrcorner"{anchor=center, pos=0.125, rotate=180}, draw=none, from=2-2, to=1-1]
    \end{tikzcd}\]
\end{enumerate}

\section{Encoding of Cubes and Simplices}
To encode the finite cubes $\I^n$ and the simplices $\Delta^n$ and $\Delta_n$, we introduced a vector type (\verb|SemiLattice.Vector.Vector|) defined inductively as
\begin{code}
Vector : Type → ℕ → Type
Vector A zero = Unit
Vector A (suc n) = A × Vector A n
\end{code}
together with a parametrised predicate that indicates a vector is monotonic with respect to a certain binary relation (\verb|SemiLattice.Vector.IsMonotonic|)
\begin{code}
IsMonotonic : ∀ {A} (R : A → A → Type) → ∀ {n} → Vector A n → Type
IsMonotonic R {0} _ = Unit*
IsMonotonic R {1} _ = Unit*
IsMonotonic R {suc (suc _)} (x , y , z) =
    (R x y) × (IsMonotonic R (y , z))
\end{code}
Notice that this predicate assumes that \verb|R| is transitive. The simplices $\Delta^n$ and $\Delta_n$ are then defined as the sum type of $\I^n$ as \verb|Vector| and a witness of the predicate \verb|IsMontonic|.

The infinite cases $\I^\N$, $\Delta^\omega$, and $\Delta^\infty$ are more complicated. In the on-paper proofs, we have a chain of subsets $\Delta^\omega\subseteq\Delta^\infty\subseteq\I^\N$. While we define $\Delta_\infty$ as subsets of $\I^\N$ satisfying a predicate (\verb|Omega.increasing∞|), the $\omega$-simplex $\Delta^\omega$ is defined as the colimit of chains (see Lemma \ref{lemma-chain}), which is a HIT in Cubical Agda (\verb|Cubical.HITs.SequentialColimit|). This allows us to utilize the results for finite cases when proving theorems for the infinite cases.

\section{Encoding of Spines}
As described earlier in Chapter \ref{chapter-phoa} and \ref{chapter-omega}, spines can be expressed as HITs in directed homotopy, which can be further translated into HITs in Cubical. The actual definition of the $\omega$-spine $\Lambda_\omega$ (\verb|Omega.Λω|) is the same as described in Chapter \ref{chapter-omega}. However, for finite cases, we have to find a way to inductively define the $n$-spines inductively for all $n\in\N$. Notice that $\Lambda_0$ is a single point, and $\Lambda_{n+1}$ can be seen as $\Lambda_n$ with an additional point attached to one of its ends via a directed path. To illustrate this in Agda, we do not give an inductive definition of $\Lambda_n$ directly. Instead, we see $\Lambda_n$ as a pointed space (\verb|Lattice.Λ∙|), which is the result of iterative application the following type constructor (\verb|Lattice.⊥∙|):
\begin{code}
data ⊥∙ (A∙ : Pointed) : Type where
    base : ⊥∙ A∙
    incl : typ A∙ → ⊥∙ A∙
    path : S → ⊥∙ A∙
    path-0 : path s0 ≡ base
    path-1 : path s1 ≡ incl (pt A∙)
\end{code}
which corresponds to the following directed HIT:
\begin{code}
data ⊥∙ (A∙ : Pointed) : Type where
    base : ⊥∙ A∙
    incl : typ A∙ → ⊥∙ A∙
    path : base ⊑ incl (pt A∙)
\end{code}
Notice that the base points of the pointed spaces are always chosen as \verb|base|.

One may ask why we choose to inductively add points to the bottom instead of to the top. The reason is that, when evaluating a function $f:\Lambda_n\to\I$ over the vertices, we will always start with the lowest one, resulting in an ascending sequence of values. This is consistent with Theorem \ref{thm-spine}: $\I^{\Lambda_\omega}\cong\Delta_\infty$. It can be shown that doing the other way will give an isomorphic definition of the spines.

One may also compare this construction with the Sierpi\'nski cone \cite{pugh_when_2025}. While both constructions introduce a new minimum, the cone construction adds paths from the new minimum to all other paths, while our construction adds only one path from the new minimum to the old and hence is only applicable to pointed spaces. The following diagrams depict the difference.
\[\begin{tikzcd}
    & \bullet &&& \bullet \\
    \bullet & \bullet && \bullet & \bullet \\
    \\
    {*} &&& {*}
    \arrow[from=2-1, to=1-2]
    \arrow[from=2-1, to=2-2]
    \arrow[from=2-2, to=1-2]
    \arrow[from=2-4, to=1-5]
    \arrow[from=2-4, to=2-5]
    \arrow[from=2-5, to=1-5]
    \arrow[dashed, from=4-1, to=1-2]
    \arrow[dashed, from=4-1, to=2-1]
    \arrow[dashed, from=4-1, to=2-2]
    \arrow[dashed, from=4-4, to=2-4]
\end{tikzcd}\]
On the left is the cone construction, which introduces three paths, while the construction on the right introduces only one.